\documentclass[12pt]{amsart}

\usepackage{amssymb, latexsym, amsmath, extarrows, mathrsfs, amsthm,color, amsrefs} 
\usepackage{graphicx}
\usepackage{xcolor}
\usepackage{MnSymbol}

\usepackage[margin=1in, centering]{geometry}
\usepackage[bookmarksnumbered, colorlinks, plainpages]{hyperref}
\usepackage{hyperref} %hypertexnames=false,colorlinks,[pagebackref]
\hypersetup{
    colorlinks=true,       % false: boxed links; true: colored links
    linkcolor=blue,          % color of internal links
    citecolor=magenta,        % color of links to bibliography
    filecolor=magenta,      % color of file links
    urlcolor=cyan           % color of external links
}

\usepackage{mathtools}
\mathtoolsset{showonlyrefs,showmanualtags}

\newtheorem{theorem}{Theorem}[section]

\newtheorem{lemma}[theorem]{Lemma}
\newtheorem{proposition}[theorem]{Proposition}
\newtheorem{remark}[theorem]{Remark}

\newtheorem{definition}{Definition}[section]
\newtheorem{corollary}[theorem]{Corollary}

\newtheorem*{acknowledgment}{Acknowledgment}

\newcommand{\e}{\varepsilon}

\newcommand{\la}{\lambda}

\newcommand{\sj}{\mathscr{J}}

\newcommand{\Z}{\mathbb Z}
\newcommand{\R}{\mathbb R}
\newcommand{\C}{\mathbb C}
\newcommand{\T}{\mathbb T}
\newcommand{\N}{\mathbb N}

\newcommand{\eqdef}{\overset{\mathrm{def}}{=\joinrel=}}

\definecolor{deepgreen}{cmyk}{1,0,1,0.5}

\title{On the Ballistic Transport for limit-periodic Jacobi Matrices}

\author[R.\ Han]{Rui Han}
\address{Department of Mathematics \\ Louisiana State University  \\  Baton Rouge, LA 70803, USA}
\email{rhan@lsu.edu}

\author[M.\ Solis]{Moises Gomez Solis}
\address{Department of Mathematics \\ Louisiana State University  \\  Baton Rouge, LA 70803, USA}
\email{mgome29@lsu.edu}

\thanks{R. Han and M.G. Solis were partially supported by NSF-DMS-2143369.}

\begin{document}

\begin{abstract}
Motivated by Open Problem 17 in Damanik-Fillman's manuscript \cite{DaFi}, we study the decay rate of limit-periodic approximation such that ballistic transport persists in the limit-periodic setting. We show an exponential decay rate, roughly $\|\sj-\sj_{q_n}\|\sim e^{-\eta q_n}$, is sufficient. This improves on earlier results of Fillman \cite{Fil2017}.
\end{abstract}

\date{\today}

\maketitle

\section{Introduction}

% a somewhat short discussion about what is the problem at hand and its relevance % 

We consider limit-periodic Jacobi matrices and their connection to quantum dynamics. When we say \textit{Jacobi matrix} we mean a bounded self-adjoint operator, $\sj = \sj(\{a_n\},\{b_n\})$, defined by
    \begin{equation}
        (\sj \phi)_n = a_{n-1}\phi_{n-1} + b_n\phi_n + a_{n}\phi_{n+1}, \quad\quad n\in \Z, \quad\quad  \text{ acting on }\phi\in\ell^2(\Z),  
    \end{equation}
with $a \equiv \{a_n\}_{n=-\infty}^\infty, b\equiv \{b_n\}_{n=-\infty}^\infty \in \ell^\infty(\Z,\R)$, and $\inf_n a_n >0$. 
We would like to understand the asymptotic behavior of solutions of the time-dependent Schr\"odinger equation, 
    \begin{equation}\label{time-dependent schrodinger}
        i \partial_t \Psi = \sj \Psi,
    \end{equation}
with initial condition $\Psi(0) = \phi_0\in \ell^2(\Z)$. As it is well-known in the literature \cites{R&S, Tesc}, we can write solutions to \eqref{time-dependent schrodinger} as 
    \begin{equation*}
        \Psi(t) = e^{-it\sj}\phi_0, \quad\quad t\in \R.
    \end{equation*}
 
These solutions have been thoroughly studied in a myriad of perspectives and there is a wealth of (classic) literature in both mathematics and physics discussing them, see \cites{Tesc2,Kat, R&S}, and the references there in. In this paper we study how fast the initial state $\phi_0$ spreads in space as time evolves.

We denote the Heisenberg evolution of the position operator relative to our Jacobi matrix $\sj$ by,
    \begin{align}
        X_{\sj} \eqdef e^{it\sj}X e^{-it\sj}, \quad\quad t\in \R. 
    \end{align}
Here $X$ is the position operator, which is defined by,
\begin{align}\label{position operator}
    (X\phi)_n &\eqdef n\phi_n, \quad\quad n\in \Z\\
    \phi \in \mathfrak{D}(X) & \eqdef \left\{\phi\in\ell^2(\Z) \colon \|X\phi\|_{\ell^2(\Z)}<\infty\right\}\subseteq \ell^1(\Z).
\end{align}

Heuristically, we want to see whether an initial state, after some time $t$, will obey a rough linear growth pattern. In terms of operators, this would mean that for our evolved position operator, we would expect $\frac{1}{t}X_{\sj}(t)$ to tend to a nonzero operator asymptotically. More formally, we say that the dynamics associated to our Jacobi matrix show \textit{ballistic transport} (in the strong sense), if there exists a bounded self-adjoint operator $Q_{\sj}$ such that $\operatorname{Ker} (Q_{\sj})\bigcap \ell^1(\Z)= \{0\}$ and 
\begin{equation}
    \text{s-}\lim_{t\rightarrow\infty} \frac{1}{t}X_{\sj} = Q_{\sj}.
\end{equation}
For the case of $1$-dimensional Jacobi matrices, there are results relating Ballistic transport to the fact that Jacobi matrices have \textit{purely absolutely continuous} (p.a.c.) spectrum \cites{DaGa,chau}.

It is known that if there exists $q\in\Z_+$ such that the Jacobi coefficients satisfy $a_{n+q}= a_n$ and $b_{n+q}=b_q$, for all $n\in \Z$, the dynamics are ballistic (in the strong sense). When the elements of a Jacobi matrix satisfy this $q$-periodicity, we simply say that the matrix $\sj$ is $q$-periodic.
Since periodic Jacobi matrices have all the properties we could expect to ensure the presence of ballistic transport it becomes natural to ask under what conditions, approximations by periodic Jacobi matrices, can maintain these dynamical properties. 

Throughout the paper we make the same assumption of $R$-boundedness on $\sj$, as in \cite{Fil2017}. That is, there exists $R>0$ such that for any $n\in \N$, $1/R\leq a_n\leq R$ and $|b_n|\leq R$. We denote the family of $R$-bounded Jacobi matrices by $\mathcal{J}(R)$.

\begin{definition}
    We say a Jacobi matrix, $\sj$, is \textit{limit-periodic} if there exists a sequence $\{\sj_{q_n}\}_{n=0}^{\infty}$ consisting of $q_n$-periodic matrices such that
    \begin{equation}
        \lim_{n\rightarrow\infty} \|\sj - \sj_{q_n}\| = 0.
    \end{equation}
\end{definition}

In \cite{DYL} Damanik-Lukiç-Yessen proved ballistic transport for periodic block-valued Jacobi matrices. Later Fillman \cite{Fil2017} showed it is possible to establish ballistic transport for limit-periodic operators if one imposes some conditions on the rate of convergence. 
More precisely, Fillman showed that for $\sj\in \mathcal{J}(R)$, there exists $\eta_0=\eta_0(R)>0$ such that if there exists a family of $q_n$-periodic Jacobi matrices $\sj_{q_n}$ satisfying
\[\lim_{n\to \infty} e^{\eta_0q_{n+1}}\|\sj-\sj_{q_n}\|=0,\]
then $\sj$ exhibits ballistic transport. 
We also mention that the same type of result was established by Young \cite{Young} in the continuum case. 
In Open Problem 17  (see, appendix E) in \cite{DaFi}, Damanik and Fillman asked what is the optimal decay rate of $\|\sj-\sj_{q_n}\|$ such that ballistic transport holds for $\sj$. 
Theorem \ref{thm:main} below makes progress on this open problem by  relaxing the convergence rate to roughly $e^{-\eta q_n}$, see definition below.

\begin{definition}\label{exponential class}
    Given $\eta>0$ we say a Jacobi matrix, say $\sj$, is limit-periodic of \textit{exponential class} $\eta$. That is, $\sj\in \text{EC}(\eta)$, if there are integers $\{q_n\}_{n=0}^\infty \subseteq \N$, with  $q_n| q_{n+1}$, $q_n\neq q_{n+1}$, such that $\sj_{q_n}$ is $q_n$-periodic and 
    \begin{equation}
        \lim_{n\rightarrow \infty } q_{n+1}^{3}e^{\eta q_n}\|\sj- \sj_{q_n}\|=0.
    \end{equation}
\end{definition}

Before we state our main theorem, let us introduce the notions of quantum transport exponents.
For $p>0$ and $\phi\in\ell^2(\Z)$ we define
\begin{equation}
    |X|^p_\phi(t) \eqdef \sum_{n\in \Z}(|n|^p+1)|\langle \delta_n, e^{-it\sj}\phi\rangle|^2.
\end{equation}
We also define the discrete Schwartz space as
\begin{equation}
    \mathcal{S}= \mathcal{S}(\Z) \eqdef \left\{\phi \in \ell^2(\Z)\left| \quad \sum_{n\in\Z}|n|^p|\phi_n|<\infty \quad \text{for all } p>0\right.\right\}.
\end{equation}
As we are primarily interested in the growth of $|X|^p_\phi(t)$, for $\phi \in \mathcal{S}$, in polynomial time we consider the (quantum) dynamical exponents
\begin{equation}
    \beta_\phi^+(p)\eqdef \limsup_{t\longrightarrow\infty}\frac{\log |X|^p_\phi(t)}{p \log t}, \quad\quad \beta_\phi^-(p)\eqdef \liminf_{t\longrightarrow\infty}\frac{\log |X|^p_\phi(t)}{p \log t}
\end{equation}
Our main theorem being as follows: %is the following:
\begin{theorem}\label{thm:main}
For every $R>0$, there is a constant $\eta=\eta(R)$ such that if $\sj\in \mathcal{J}(R)$ and further $\sj\in \text{EC}(\eta)$, namely
\begin{align}\label{eq:assume}
\lim_{n\rightarrow \infty } q_{n+1}^{3}e^{\eta q_n}\|\sj- \sj_{q_n}\|=0,
\end{align}
then the dynamics is ballistic in the sense that there exists a self-adjoint operator $Q_\sj$ such that 
\begin{equation}\label{eq: strong}
    \lim_{t\longrightarrow \infty} \frac{1}{t}X_{\sj}= Q_{\sj},
\end{equation}
with $\operatorname{Ker}(Q_\sj)\bigcap \ell^1(\Z)={0}$. In particular, $\beta_{\phi}^+(p) = 1$ for all $\phi \in \mathcal{S}$ and all $p>0$. 
\end{theorem}
In fact, the result holds for any $\eta>7c_1$, where $c_1=c_1(R)$ is as in Lemma \ref{lem:Last}. 
From here onward we are going to be assuming that $\eta>7c_1$.

\begin{remark}
    In the literature, the dynamics with $\beta^+(p)=1$, for any $p>0$, are also called quasiballistic. Quasiballistic transports have also been studied in the quasi-periodic setting in the regime of purely absolutely continuous spectrum \cite{Ka,ZZ,Zh}, or with Liouville frequencies \cite{La2,JZ}. We briefly comment on some connections to Liouville frequencies in Sec.\ref{sec:QP}. 
\end{remark}

\begin{remark}
    One might ask if a lower bound $\inf a_n\geq 1/R>0$ is necessary. If for some $n$, $a_n=0$ in the periodic case with period $q$, then using $a_n=a_{n+q}=0$, one can construct finite volume eigenfunction $\phi$ in the box $[n+1,n+q]$, for which the dynamics $e^{-it\sj}\phi$ does not propagate as time evolves, so ballistic transport does not hold.
\end{remark}

\subsection*{Structure of the paper}
 In Section 2 we go over some of the properties of Jacobi matrices belonging to $\text{EC}(\eta)$, for $\eta>0$. In particular we introduce as lemmas some known results estimating the size of the bad set $\mathcal{B}_{q_n,\varepsilon}$. In Section 3 we show how to obtain an improved estimate for $\mathcal{B}_{q_n,\varepsilon}$. In Section 4 we introduce estimates for the required rate of convergence, and prove the main Theorem \ref{thm:main} in Section 5. Some connections to the quasi-periodic setting is discussed briefly in Section 6.

\section{Preliminaries}

\subsection{Limit-periodic setup}

For notational convenience we introduce the following
    \begin{equation}
        \mathcal{J}_{q_n}(R) \eqdef \left\{\sj\in \mathcal{J}(R) \colon \sj \text{ is  } q_{n}\text{-periodic}\right\}.
    \end{equation}
    
Throughout this paper we require $\sj\in \mathcal{J}(R)\cap EC(\eta)$. It is clear that for $n$ large enough, $\sj_{q_n}\in \mathcal{J}_{q_n}(2R)$. We can rescale $R$, and for simplicity just assume $\sj\in \mathcal{J}(R)$ and $\sj_{q_n}\in \mathcal{J}_{q_n}(R)$ for $n$ large enough.

\begin{proposition}\label{decomposition}
    A Jacobi matrix $\sj$ can be decomposed in the form:
    \begin{equation}\label{Jacobi series form}
        \sj= \sum_{k=1}^\infty H_k,  \quad H_k = \begin{cases}
                \sj_{q_1}, & k=1\\
                \sj_{q_{k}}-\sj_{q_{k-1}}, & k\geq 2
            \end{cases}, \quad H_k\in \mathcal{P}(\ell^2(\Z); q_k).
    \end{equation}
    Here $\mathcal{P}(\ell^2(\Z); q_k)$ is defined to be the space of $q_k$-periodic operators in $\ell^2(\Z)$; These operators, for $k$ large enough satisfy $\|H_k\| \leq e^{-\eta q_{k-1}}/q_{k}^3$.
    \begin{proof}
        Since $\sj\in \text{EC}(\eta)$, by definition there exists a sequence of integers $\{q_n\}_{n=0}^\infty \subseteq \N$, with $q_n| q_{n+1}$, $q_n\neq q_{n+1}$, such that 
        \begin{equation}
            \lim_{n\rightarrow \infty } q_{n+1}^{3}e^{\eta q_n}\|\sj- \sj_{q_n}\|=0, \qquad \forall n \in \N.
        \end{equation}
        Moreover, for $n$ large enough, we have $\|\sj-\sj_{q_n}\|< ce^{-\eta q_{n}}/q_{n+1}^3$ for any small $c>0$. As such, we can define
        \begin{equation}
            H_k = \begin{cases}
                \sj_{q_1}, & k=1\\
                \sj_{q_{k}}-\sj_{q_{k-1}}, & k\geq 2
            \end{cases}.
        \end{equation}
        To see that $H_k$ is $q_k$-periodic, notice that $\sj_{q_k}$ is by definition $q_k$-periodic and $\sj_{q_{k-1}}$ can be taken as $q_k$-periodic since $q_{k-1}|q_{k}$. Furthermore,
        \begin{equation}
            \lim_{N\rightarrow\infty} \sum_{k=1}^N H_k = \lim_{N\rightarrow\infty} \sum_{k=2}^N (\sj_{q_{k}}-\sj_{q_{k-1}}) + \sj_{q_1} = \lim_{N\rightarrow\infty}  \sj_{q_N}=\sj,
        \end{equation}
        and for large enough $k$
        \begin{equation}
            \|H_k\| \leq \|\sj-\sj_{q_k}\| +\|\sj-\sj_{q_{k-1}}\| \leq \frac{ce^{-\eta q_{k}}}{q_{k+1}^3} + \frac{ce^{-\eta q_{k-1}}}{q_{k}^3} \leq \frac{e^{-\eta q_{k-1}}}{q_{k}^3},
        \end{equation}
        as claimed.
    \end{proof}
\end{proposition}

\subsection{Bad sets $\mathcal{B}_{q_n,\varepsilon}$}
In order to prove ballistic transport, we proceed similarly as in \cites{DYL,Fil2017}, that is, we want to determine an upper bound for the measure of ``bad'' sets $\mathcal{B}_{q_n,\e}$, where 
\begin{equation}\label{bad set (generic)}
    \mathcal{B}_{q_n,\e} \eqdef \left\{ \theta \in \T \colon  \, \exists\,  0\leq j < k\leq q_n-1 \text{ s.t. } |\lambda_{q_n, j}(\theta) - \lambda_{q_n, k}(\theta)|\leq \e \right\}.
\end{equation}
Here the $\lambda_{q_n, i}(\theta)$'s are the eigenvalues of a given $\sj_{q_{n}}(\theta)$, where 
\begin{equation}
    \sj_{q_n}(\theta) = \begin{pmatrix}
        b_1 & a_1 &  & &   a_{q_n}e^{i2\pi\theta}\\
        a_1 & b_2 &  & &  \\
         & & \ddots& & \\
         & & & b_{q_{n-1}} & a_{q_{n-1}} \\
        a_{q_n}e^{-i2\pi\theta} & & & a_{q_{n-1}} & b_{q_n}  \\
    \end{pmatrix}
\end{equation}
is the operator with pseudo-periodic boundary condition parametrized by $\theta$. From \textit{Floquet theory} we know that the spectrum of $\sj_{q_{n}}\in\mathcal{J}_{q_{n}}(R)$ is such that
\begin{equation}
    \sigma\left(\sj_{q_{n}}\right)= \bigcup_{\theta \in \T}\sigma\left(\sj_{q_{n}}(\theta)\right), \qquad \T\simeq [0,1)
\end{equation}

\begin{theorem}\cite[Theorem 5.3.4]{Sim} \label{theorem: properties of the eigenvalues}
Let $\{\lambda_{_{q_n},k}(\theta)\}_{k=0}^{q_n -1}$ be the eigenvalues of ${\sj_{q_n}(\theta)}$, then
\begin{enumerate}%[label = {\roman*)}]
    %\item We have, after rearrangement, 
   %\begin{equation*}
        %\lambda_{q_{n},q_{n}-1}(0 ) > \lambda_{q_{n},q_{n}-1}\left(\frac{1}{2} \right) \geq \lambda_{q_{n},q_{n}-2}\left(\frac{1}{2} \right) > \lambda_{q_{n},q_{n}-2}(0 )\geq \dots 
   % \end{equation*}
    \item $\lambda_{_{q_n},k}(\theta)$ is analytic for $\theta \in (0, 1/2)\cup (1/2, 1)$.
    \item $\lambda_{{q_n},k}(\theta)$ is strictly monotone in $(0,1/2)$ and $(1/2,1)$.
\end{enumerate}
    
\end{theorem}

\subsection{Last's result}
Clearly the measure of bad sets can be controlled via certain lower bounds on the slopes of eigenvalues. 
In this section we elaborate on some necessary estimations on the slopes, obtained by Last \cite{Last} in the setting of discrete Schr\"odinger operators, which are going to be cornerstone for later discussion on both, the proof of the perturbation step and proving ballistic transport. We present the result in the form of lemmas as to make the whole more self-contained.

\begin{lemma}\label{lem:Last}
    For every $R>0$ there is a constant $c_1 = c_1(R)>0$, such that, for any  $\sj\in \mathcal{J}_{q_n}(R)$, with associated $\sj_{q_n}(\theta)$, one has. 
    \begin{equation}\label{Lats's result}
        |\Dot{\lambda}_{q_n,j}(\theta)| \geq e^{-c_1q_n}|\sin2 \pi\theta|, \quad \quad \forall 1\leq j \leq q_n \text{ and } \quad \forall \theta\in \T.
    \end{equation}
    \begin{proof}
        Let $R>0$ and 
        \begin{equation}
            \sj_{q_n}\left(\theta\right) = \begin{pmatrix}
        b1 & a_1 & 0 & &  \cdots &e^{i 2\pi\theta}a_{q_n} \\
        a_1 & b_2 & a_2  & & & \\
         0& a_2& b_3& \ddots & & \\
         \vdots & & \ddots & \ddots&  \ddots&\\
          & & & \ddots & b_{q_{n} -1} & a_{q_{n}-1} \\
         e^{i 2\pi\theta}a_{q_n}& & & & a_{q_{n}-1} & b_{q_{n}}  \\
            \end{pmatrix}.
        \end{equation}
        we see that 
        \begin{align*}
            \operatorname{det}(\lambda - \sj_{q_n}(\theta)) & = (a_1a_2\cdots a_{q_n})\left(\Delta(\lambda)- 2\cos{2\pi\theta}\right), 
        \end{align*}
        here $\Delta(\lambda)$ denotes the ``discriminant'' (see chapter 5 in \cite{Sim}) or the trace of the transfer matrix induced by $\sj_{q_n}(\theta)$, and is a polynomial of degree $q_n$. Since the eigenvalues, $\lambda_{q_n,j}(\theta)$, of $\sj_{q_n}(\theta)$ solve
        \begin{equation*}
            \Delta(\lambda_{q_n,j}(\theta)) = 2 \cos{2\pi\theta},
        \end{equation*}
        we have
        \begin{equation*}
            \dot\Delta(\lambda_{q_n,j}(\theta))\cdot \dot{\lambda}_{q_n,j}(\theta) = - 4\pi\sin{2\pi\theta}.
        \end{equation*}
        Notice that $\Delta'(\lambda_{q_n,j}(\theta))$ is a polynomial of degree $q_n-1$ whose zeros separate those of $\Delta(\lambda_{q_n,j}(\theta))$, or equivalently, the eigenvalues of $\sj(\theta)$. Thus, 
        \begin{align*}
            \left|\dot{\lambda}_{q_n,j}(\theta)\right| & = \left|\frac{2\sin{2\pi\theta}}{ \dot\Delta(\lambda_{q_n,j}(\theta))}\right|\\
            & \geq e^{-c_1q_n}|\sin{2\pi\theta}|,
        \end{align*}
        where $c_1>0$ depends on the $a_j$'s, and hence on $R$.
    \end{proof}
\end{lemma}

As a corollary of Last's estimates, Fillman obtained the following upper bound on $\mathrm{mes}(B_{q_n,\varepsilon})$ in \cite{Fil2017}.
\begin{lemma}\label{Lemma: with Last's estimate}
    For $\varepsilon>0$,
    \begin{equation}
        \operatorname{mes}\left({\mathcal{B}}_{q_n,\e}\right) \leq 4e^{c_1q_n/2}\sqrt{\e}.
    \end{equation}
    \begin{proof}
    Assume that $\theta\in \mathcal{B}_{q_n,\varepsilon}$, then there exists $0\leq j<k\leq q_n-1$ such that 
        \begin{align}\label{eq:k<j}
            \lambda_{q_n,k}(\theta)& \leq \e +\lambda_{q_n,j}(\theta)
        \end{align}
Now we use the general properties of 1D eigenvalues in Theorem \ref{theorem: properties of the eigenvalues} that the spectral bands do not overlap, meaning for $j<k$,
\begin{align}
    \max_{\theta} \lambda_{q_n,j}(\theta)\leq \min_{\theta}\lambda_{q_n,k}(\theta).
\end{align}
Combining this with \eqref{eq:k<j} implies
        \begin{align}
            \lambda_{q_n,k}(\theta)& \leq \e + \min_{\theta} \lambda_{q_n,k}(\theta).
        \end{align}
        Assume without loss of generality $\min_{\theta}\lambda_{q_n,k}(\theta)=\lambda_{q_n,k}(0)$ (the case when minimum is attained at $1/2$ is completely analogous).
        By Lemma \ref{lem:Last},
        \begin{align}
            \varepsilon\geq \lambda_{q_n,k}(\theta)-\lambda_{q_n,k}(0)& =\int_{0}^{\theta}\dot\lambda_{q_n,k}(s)\text{d}s\\
            & \geq 4\int_0^\theta e^{-c_1q_n}\|s\|_{\T}\text{d}s\\
            & \geq 2e^{-c_1q_n}\theta^2.
        \end{align}
        This implies $|\theta|\leq e^{c_1q_n/2}\sqrt{\frac{1}{2}}$. Hence $\mathcal{B}_{q_n,\varepsilon}$ must be contained in neighborhoods of $0$ and $1/2$ with total measure $<4 e^{c_1q_n/2}\sqrt{\varepsilon}$. This completes the proof.
    \end{proof}
\end{lemma}
This estimate holds for any $\varepsilon>0$ and it holds for any $q_n$-periodic $\sj_{q_n}\in \mathcal{J}_{q_n}(R)$.
In the next section, we show an improved bound for $\mathrm{mes}(\mathcal{B}_{q_n,\varepsilon})$ for $\varepsilon$ not too small, by using the fact that $\sj_{q_n}$ is a small perturbation of a $q_{n-1}$-periodic matrix $\sj_{q_{n-1}}$, due to Proposition \ref{decomposition} that $\sj_{q_n}=\sj_{q_{n-1}}+H_n$ with $\|H_n\|\leq e^{-\eta q_{n-1}}/q_n^3$.
We present this improvement in the next section.

\section{An improved estimate for $\mathrm{mes}(B_{q_n,\varepsilon})$}
As a preparation we have the following proposition.
\begin{proposition}\label{proposition: distance of eigenvalues}

    Let $\sj\in\text{EC}(\eta)\cap \mathcal{J}(R)$, and  $\sj_{q_n}\in \mathcal{J}_{q_n}(R)$ be its $q_n$-periodic approximant. Denote by $\tilde{\sj}_{q_{n}} \eqdef \sj_{q_{n-1}}$, that is, the $q_{n-1}$-periodic Jacobi matrix \textit{artificially} seen as $q_{n}$-periodic Jacobi matrix. Let $\sj_{q_{n}}(\theta)$ and $\tilde\sj_{q_{n}}(\theta)$ be the Floquet matrices of $\sj_{q_{n}}$ and $\tilde\sj_{q_{n}}$, respectively. Denote by $\sigma(\sj_{q_{n}}(\theta)) = \{\lambda_{q_{n},k}(\theta)\}_{k=0}^{q_{n}-1}$ and $\sigma(\tilde \sj_{q_{n}}(\theta)) = \{\tilde\lambda_{q_{n},k}(\theta)\}_{k=0}^{q_{n}-1}$, the sets of \textit{ordered} eigenvalues. In other words, we have
    \begin{equation}
            \tilde{\lambda}_{q_{n},k}(\theta) \leq \tilde{\lambda}_{q_{n},k+1}(\theta), \quad \text{and} \quad \lambda_{q_{n},k}(\theta) \leq \lambda_{q_{n},k+1}(\theta) \qquad 0\leq k \leq q_{n}-1. 
        \end{equation}
    Then
    \begin{equation}
        |\tilde{\lambda}_{q_{n},k}(\theta) - \lambda_{q_{n},k}(\theta)| \leq e^{-\eta q_{n-1}}/q_{n}^3, \quad \text{ for  each } 0 \leq k \leq q_{n}-1.
    \end{equation}
    
    \begin{proof}
        Since for $n$ large enough we have $\|\sj_{q_{n}}-\tilde\sj_{q_{n}}\|= \|\sj_{q_{n}}-\sj_{q_{n-1}}\| = \|H_{q_{n}}\|\leq e^{-\eta q_{n-1}}/q_{n}^3$, by proposition \ref{decomposition}. The same estimate works for the Floquet matrices. In other words, $\|\sj_{q_{n}}(\theta)-\tilde\sj_{q_{n}}(\theta)\|\leq e^{-\eta q_{n-1}}/q_{n}^3$ uniformly in $\theta$.
        Thus the result follows readily from the minimax principle (see \cite{Kat}.)
    \end{proof}
\end{proposition}

We now proceed to estimate the upper-bound of \eqref{bad set (generic)}. Moreover, from now to the rest of the manuscript we introduce for convenience the following: $\gamma_{n} \eqdef e^{-\eta q_{n-1}}/q_n^3$.

\begin{theorem}\label{theorem : estimate for bad sets}
    Let $0\leq  \gamma_n < \e$. Then, 
    \begin{equation}
        \operatorname{mes}\left(\mathcal{B}_{q_{n},\e}\right)\leq q_{n}e^{c_1q_{n-1}/2}\sqrt{\e}.
    \end{equation}
\end{theorem}
    \begin{proof}
        Consider 
        \begin{equation}
            \mathcal{B}_{q_{n},\e} = \left\{ \theta \in \T \colon  \, \exists \,0\leq j < k\leq q_{n}-1 \text{ s.t. } |\lambda_{q_{n}, j}(\theta) - \lambda_{q_{n}, k}(\theta)|\leq \e_{} \right\}.
        \end{equation}
        By proposition \ref{proposition: distance of eigenvalues} we must have
        \begin{equation*}
            \mathcal{B}_{q_{n},\e} \subseteq \left\{ \theta \in \T \colon  \, \exists \,0\leq j < k\leq q_{n}-1 \text{ s.t. } |\tilde\lambda_{q_{n}, j}(\theta) - \tilde\lambda_{q_{n}, k}(\theta)|\leq \e + 2\gamma_{n}\right\} =\colon \Tilde{\mathcal{B}}_{q_{n},\e_{} + 2\gamma_{n}}.
        \end{equation*}
        Since we have by hypothesis $\gamma_{n}<\e$, it follows that $\Tilde{\mathcal{B}}_{q_{n},\e + 2\gamma_{n}}\subseteq \Tilde{\mathcal{B}}_{q_{n},3\e}$. It is well-known (see e.g. \cite{HPS}) that the two sets of eigenvalues $\{\tilde{\lambda}_{q_{n}, k}(\theta)\}$ and $\{\lambda_{q_{n-1}, k}(\theta)\}$ are related in the sense 
        \begin{equation}
            \bigcup_{k=0}^{q_{n}-1}\left\{\tilde{\lambda}_{q_{n}, k}(\theta)\right\} = \bigcup_{k=0}^{q_{n-1}-1}\bigcup_{l=0}^{l_{n-1}-1}\left\{\lambda_{q_{n-1}, k}\left(\frac{\theta+l}{l_{n-1}}\right)\right\},\qquad \text{where } l_{n-1}=\frac{q_{n}}{q_{n-1}} .
        \end{equation}
        Recall that Last's result \eqref{Lats's result} implies
        \begin{equation}
            \left|{\dot\lambda}_{q_{n-1}, k}(\theta)\right| \geq e^{-c_1q_{n-1}}|\sin{2\pi\theta}|,
        \end{equation}
        and hence if $\tilde{\lambda}_{q_n,k}(\theta)=\lambda_{q_{n-1},\tilde{k}}(\frac{\theta+\ell}{\ell_{n-1}})$, then
        \begin{align}\label{eq:tlambda_derivative}
            \left|{\dot{\tilde\lambda}}_{q_{n}, k}(\theta)\right| &= \frac{1}{l_{n-1}}\left|\dot{{\lambda}}_{q_{n-1}, \tilde{k}}\left(\frac{\theta+l}{l_{n-1}}\right)\right|
             \geq \frac{1}{l_{n-1}}e^{-c_1q_{n-1}}\left|\sin{2\pi\left(\frac{\theta+l}{l_{n-1}}\right)}\right| \notag\\
            & \geq \frac{c}{l_{n-1}}e^{-c_1q_{n-1}} \inf_{l}\left\|\left(\frac{2\theta+2l}{l_{n-1}}\right)\right\|_\T \notag\\
             &\geq \frac{c}{l_{n-1}}e^{-c_1q_{n-1}}  \inf_{r\in \N}\inf_{l}\left|\left(\frac{2\theta+2l + rl_{n-1}}{l_{n-1}}\right)\right| \notag\\
            & = \frac{c}{l_{n-1}^2}e^{-c_1q_{n-1}} \inf_{r\in \N}\inf_{l}\left|\left(2\theta+2l + rl_{n-1}\right)\right|
             \notag \\
            & = \frac{c}{l_{n-1}^2}e^{-c_1q_{n-1}} \left\|2\theta\right\|_{\T}.
        \end{align}
        Here $\|\cdot\|_{\T} = \operatorname{dist}(\cdot, \Z)$ and the absolute constant $c>0$ is such that $|\sin \pi \phi| \geq c\|\phi\|_{\T}$.
        
        Henceforth, if $\theta\in \tilde{B}_{q_n,3\varepsilon}$, then we have for some $j<k$, 
        \begin{equation}
            |\tilde\lambda_{q_{n},j}(\theta)-\tilde\lambda_{q_{n},k}(\theta)|\leq 3\e.
        \end{equation}
        Then similar to the proof of \ref{Lemma: with Last's estimate}, we have
        \begin{equation}
            \tilde\lambda_{q_{n},k}(\theta) \leq \tilde\lambda_{q_{n},j}(\theta) + 3\e \leq \max_{\theta} \tilde\lambda_{q_{n},j}(\theta) + 3\e \leq \min_{\theta} \tilde\lambda_{q_{n},k}(\theta) + 3\e.
        \end{equation}
        Assume without loss of generality $\min_{\theta}\tilde{\lambda}_{q_n,k}(\theta)=\tilde{\lambda}_{q_n,k}(0)$. 
        Then, by \eqref{eq:tlambda_derivative},
        \begin{align}
            3\varepsilon\geq \tilde{\lambda}_{q_{n},k}(\theta)-\tilde{\lambda}_{q_{n},k}(0)& = \int_{0}^\theta \dot{\tilde{\lambda}}_{q_{n},k}(s)\text{d}s\\
            & \geq \int^\theta_0 \frac{c}{l_{n-1}^{2}}e^{-c_1q_{n-1}}\|2s\|_{\T}\text{d}s\\
            & =  \int^\theta_0 \frac{c}{l_{n-1}^{2}}e^{-c_1q_{n-1}}2s\text{d}s \qquad (\theta \text{ is small})\\
            & = \frac{c}{l_{n-1}^{2}}e^{-c_1q_{n-1}}\theta^2,
        \end{align}
        implying $|\theta|\leq C\ell_{n-1}e^{c_1q_{n-1}/2}\sqrt{\e}$.
        Therefore $\Tilde{\mathcal{B}}_{q_{n}, 3\e}$ must be contained within neighborhoods, of size $\sim l_{n-1}^{}e^{c_1q_{n-1}/2}\sqrt{\e}$, around $\{0,\frac{1}{2}\}$. Thus,
        \begin{equation}
            \operatorname{mes}(\Tilde{\mathcal{B}}_{q_{n}, 3\e}) \leq C l_{n-1}e^{c_1q_{n-1}/2}\sqrt{\e} \leq C \frac{q_{n}}{q_{n-1}}e^{c_1q_{n-1}/2}\sqrt{\e} \leq q_{n}e^{c_1q_{n-1}/2}\sqrt{\e},
        \end{equation}
        as claimed.
    \end{proof}

Having obtained this estimate, we now proceed in the same fashion as in \cite{Fil2017} and show that ballistic transport for this larger class of limit-periodic Jacobi matrices.

\section{Proof of ballistic transport}
Let $\{\lambda_{q_n,k}(\theta)\}_{k=0}^{q_n-1}$ the eigenvalues of $\sj_{q_n}(\theta)$, and by $P_{k}(\theta)$ the projection onto corresponding eigenspace: $\operatorname{Ker}(\sj_{q_n}(\theta)-\lambda_{q_n,k}(\theta)\mathbb{I})$. Following \cite{DYL}, identify
    \begin{equation}\label{The Q operator}
        Q_{q_n}(\theta ) = q_n \sum_{k=0}^{q_n-1}\dot{\lambda}_{q_n,k}(\theta)P_{k}(\theta), \qquad \theta\in \T\setminus\left\{0,\frac{1}{2}\right\};
    \end{equation}
Here we can think of $Q_{q_n}(\theta)$ as the operator acting on a copy of $\C^{q_n}$ for each $\theta\in \T \simeq [0,1)$. Let 
    \begin{equation}
    Q_{q_n} = \mathcal{F}_{q_n}^{-1}\, \left(\int_0^1 Q_{q_n}(\theta)\text{d}\theta\right) \, \mathcal{F}_{q_n}   
    \end{equation}
We also introduce $A_{q_n} (\theta)= (\mathcal{F}_{q_n} A\mathcal{F}_{q_n}^{-1})(\theta)$, where $A \eqdef i[\sj,X]$, $\mathcal{F}_{q_n}$ is the $q_n$-periodic Floquet transform and $X$ the position operator. For more details, see Appendix \ref{appendix: B}.

\begin{theorem}\label{thm:Qn-intA}
    For $R>0$, let $c_1(R)>0$ be the constant in Lemma \ref{lem:Last}.
Let $\sj\in EC(\eta)\cap \mathcal{J}(R)$ with $\eta>7c_1(R)$. Then for $n$ large, 
\begin{align}
    \left\|Q_{q_n }- \frac{1}{t}\int_0^te^{is\sj_{q_n}}A_{q_n}e^{-is\sj_{q_n}}\text{d}s\right\|^2 \leq 
        \int^{1}_0\left\|Q_{q_n}(\theta)- \frac{1}{t}\int_0^t e^{is \sj_{q_n}(\theta)}A_{q_n}(\theta)e^{-is \sj_{q_n}(\theta)}\text{d}s\right\|_{HS}^2\text{d}\theta
\end{align}
\begin{align}
     \leq 
        \begin{cases} 
        \frac{4R^2 q_n^2}{t^2\e^2} + 64R^2 e^{c_1q_{n}/2}\sqrt{\e}, \text{ for any } \varepsilon\geq 0, \text{ (Estimate I)},\\
        \frac{4R^2 q_n^2}{t^2\e^2} + 16R^2 q_ne^{c_1q_{n-1}/2}\sqrt{\e}, \text{ for any } \gamma_n<\varepsilon, \text{ (Estimate II)}.
        \end{cases}
\end{align}\label{Big Q}
\end{theorem}
     \begin{proof}
     Denote the eigenvector of $\sj_{q_n}(\theta)$ associated to eigenvalue $\lambda_{q_n,k}(\theta)$ by $\phi_k(\theta)$.
        Consider the matrix elements (see Appendix \ref{appendix: Floquet Theory}), 
        \begin{equation*}
            \left\langle \phi_j(\theta), \left(Q_{q_n}(\theta)- \frac{1}{t}\int_0^t e^{is \sj_{q_n}(\theta)}A_{q_n}(\theta)e^{-is \sj_{q_n}(\theta)}\text{d}s\right)\phi_k(\theta)\right\rangle, \quad 1\leq j,k \leq q_n, \quad \theta\in \T\setminus\{0,\pi\}.
        \end{equation*}
        In the case when $j=k$ we have 
        \begin{equation}\label{eq:diagonal}
            \left\langle \phi_k(\theta), \left(Q_{q_n}(\theta)- \frac{1}{t}\int_0^t e^{is \sj_{q_n}(\theta)}A_{q_n}(\theta)e^{-is \sj_{q_n}(\theta)}\text{d}s\right)\phi_k(\theta)\right\rangle = 0, \quad \quad \forall t\in \R.
        \end{equation}
        To see this, first observe that 
        \begin{align}
            \left\langle \phi_k(\theta), \frac{1}{t}\int_0^t e^{is\sj_{q_n}(\theta)}A_{q_n}(\theta)e^{-is\sj_{q_n}(\theta)} \phi_k(\theta)\right\rangle=\langle \phi_k(\theta),A_{q_n}(\theta) \phi_k(\theta)\rangle.
        \end{align}
        Then by \cite[Equation (10)]{DYL}, $P_k(\theta)A_{q_n}(\theta)P_k(\theta)=q_n \dot{\lambda}_{q_n,k}(\theta)P_k(\theta)$, hence the claim.
        
        For the case when $j\neq k$ and for all $\theta \in \T\setminus\{0,\frac{1}{2}\}$, we have
        \begin{align}
            \left\langle \phi_j(\theta), Q_{q_n}(\theta)\phi_k(\theta)\right\rangle=0,
        \end{align}
        and
        \begin{align}
        &\left|\frac{1}{t}\int_0^t \left\langle \phi_j(\theta), e^{is\sj_{q_n}(\theta)}A_{q_n}(\theta)e^{-is\sj_{q_n}(\theta)}\phi_k(\theta)\right\rangle\, \mathrm{d}s \right|\\
        &\qquad\qquad=\left|\frac{1}{t}\int_0^t e^{is(\lambda_{q_n,j}(\theta)-\lambda_{q_n,k}(\theta)}\, \mathrm{d}s\, \langle \phi_j(\theta), A_{q_n}(\theta)\phi_k(\theta)\rangle \right|\\
        &\qquad\qquad \leq \frac{2R}{t|\lambda_{q_n,j}(\theta)-\lambda_{q_n,k}(\theta)|},
        \end{align}
        in which we bounded $|\langle \phi_j(\theta), A_{q_n}(\theta)\phi_k(\theta)|\leq 2R$, see \eqref{eq:Aqn_theta}.
        Therefore,
        \begin{equation}\label{big denominator inequality}
            \left|\left\langle \phi_j(\theta), \left(Q_{q_n}(\theta)- \frac{1}{t}\int_0^t e^{is \sj_{q_n}(\theta)}A_{q_n}(\theta)e^{-is \sj_{q_n}(\theta)}\text{d}s\right)\phi_k(\theta)\right\rangle\right| \leq \frac{2R}{t|\lambda_{q_n,j}(\theta)-\lambda_{q_n,k}(\theta)|}.
        \end{equation}
        We want to partition $\T$ into bad and good sets, where by ``bad'' we mean sets such that the denominator in \eqref{big denominator inequality} is smaller than some prescribed $\e>0$, and thus make the estimate useless in \eqref{big denominator inequality}. Recall our bad set $\mathcal{B}_{q_n,\varepsilon}$ was defined in \eqref{bad set (generic)}. We denote the complement as $\mathcal{G}_{q_n, \e}\eqdef \T \setminus\mathcal{B}_{q_n,\e}$. Next, we proceed by cases, controlling for the size of $\e>0$.
        For simplicity, we denote 
        \begin{align}
            Q_{q_n}(\theta)- \frac{1}{t}\int_0^t e^{is \sj_{q_n}(\theta)}A_{q_n}(\theta)e^{-is \sj_{q_n}(\theta)}\text{d}s\eqdef Y(\theta)
        \end{align}
        Note by the vanishing diagonal and off-diagonal estimates in \eqref{big denominator inequality}, we have
        \begin{align}
            \|Y(\theta)\|_{HS}^2\leq \frac{4R^2 q_n^2}{t_n^2 \min\limits_{j\neq k}\left\{|\lambda_{q_n,j}(\theta)-\lambda_{q_n,k}(\theta)|^2\right\}}
        \end{align}
        \textit{\textbf{\underline{Estimate I}:}} Using Lemma \ref{Lemma: with Last's estimate},
        \begin{equation}\label{eq:Leb1}
        \operatorname{Leb}(\mathcal{B}_{q_n,\varepsilon}) \leq 4e^{c_1q_n/2}\sqrt{\e}.
    \end{equation}
    Thus,
     \begin{align}
        \left\|Q_{q_n} - \frac{1}{t_n}\int^{t_n}_0e^{is\sj_{q_n}}A_{q_n}e^{-is\sj_{q_n}}\text{d}s\right\|^2 
         &\leq \int\limits_{\mathcal{G}_{q_n,\e}}\left\|Y(\theta) \right\|_{HS}^2 \text{d}\theta + \int\limits_{\mathcal{B}_{q_n,\e}}\left\| Y(\theta)\right\|_{HS}^2\text{d}\theta\\
         &\leq \frac{4R^2q_n^2}{t_n^2 \e^2} +  16R^2 \cdot\operatorname{Leb}(\mathcal{B}_{q_n,\varepsilon})\\
         & \leq \frac{4R^2 q_n^2}{t^2_n\e^{2}} + 64 R^2 e^{c_1q_n/2} \sqrt{\e_{}},
    \end{align}
    here the integral on the bad sets follows from the measure estimate of $\mathcal{B}_{q_n,\varepsilon}$ in \eqref{eq:Leb1} and a trivial bound on $\|Y(\theta)\|_{HS}$, see \cite[Equation 2.13]{Fil2017}.
    
    \textit{\textbf{\underline{Estimate II}:}} For $0\leq  \gamma_n < \e$, we use Theorem \ref{theorem : estimate for bad sets} and observe, 
    \begin{align}\label{eq:Leb2}
        \mathrm{Leb}(\mathcal{B}_{q_n,\varepsilon})\leq q_ne^{c_1q_{n-1}/2}\sqrt{\varepsilon}.
    \end{align}
    Then similar computations yield
    \begin{align}
        \left\|Q_{q_n} - \frac{1}{t_n}\int^{t_n}_0e^{is\sj_{q_n}}A_{q_n}e^{-is\sj_{q_n}}\text{d}s\right\|^2
        \leq
        \frac{4R^2 q_n^2}{t^2_n\e^{2}} + 16 R^2q_n e^{c_1q_{n-1}/2} \sqrt{\e_{}},
    \end{align}
    where the estimate for the bad set measure is provided by \eqref{eq:Leb2}.
\end{proof}

\begin{corollary}\label{cor:1}
     Let $\varphi \in \ell^{1}(\Z)\cap \mathfrak{D}(X)$, then for $\varepsilon>\gamma_n$,
        \begin{equation}
        \left\|Q_{q_n}\varphi - \frac{1}{t} X_{\sj_{q_n}}(t)\varphi\right\| \leq \frac{1}{t}\|X\varphi\| + \|\varphi\|_{\ell^{1}}\left(\frac{4R^2 q_n^2}{t^2\e^2} + 16 R^2 q_ne^{c_1q_{n-1}/2}\sqrt{\e}\right)^{1/2},
        \end{equation}
        where $X_{\sj_{q_n}}(t) \eqdef e^{it\sj_{q_n}}Xe^{-it\sj_{q_n}}$ is the Heisenberg evolution.
        \begin{proof}
            Recall that $\mathfrak{D}(X_{\sj_{q_n}}(t)) = \mathfrak{D}(X)\subseteq \ell^1(\Z)$ for all $t\in\R$, and
            \begin{equation}\label{position identity}
                X_{\sj_{q_n}}(t) = X +\int_0^te^{is\sj_{q_n}}A_{q_n}e^{-is\sj_{q_n}}\text{d}s,
            \end{equation}
            here we used that $A_{q_n}= i[\sj_{q_n},X]$.
            Therefore
            \begin{align}
                \left\|Q_{q_n}\varphi-\frac{1}{t}X_{\sj_{q_n}}(t)\varphi\right\|
                \leq &\frac{1}{t}\|X\varphi\|+\left\|\left(Q_{q_n}-\frac{1}{t}\int_0^t e^{is\sj_{q_n}}A_{q_n}e^{-is\sj_{q_n}}\, \mathrm{d}s \right)\varphi\right\|\\
                \leq &\frac{1}{t}\|X\varphi\|+\left\|Q_{q_n}-\frac{1}{t}\int_0^t e^{is\sj_{q_n}}A_{q_n}e^{-is\sj_{q_n}}\, \mathrm{d}s\right\|\cdot \|\varphi\|_1.
            \end{align}
            Then second term above can be bounded via Estimate II in Theorem \ref{thm:Qn-intA} since we required $\varepsilon>\gamma_n$.
        \end{proof}
\end{corollary}

\section{Ballistic dynamics: proof of Theorem \ref{thm:main}}
\subsection{Strong convergence of $Q_{q_n}$}
    For the existence of such $Q_\sj$ it is enough to show that this holds for every $\phi\in\ell^1(\Z)$ satisfying 
    \begin{equation}\label{eq: Q is strong summable}
        \sum_{n=1}^\infty \|Q_{q_n}\phi - Q_{q_{n+1}}\phi\|<\infty.
    \end{equation}
   By Estimate I of Theorem \ref{thm:Qn-intA} applied to $Q_{q_n}$ and $t=t_n$ and $\varepsilon>0$ (values to be determined below), we have
   \begin{align}\label{eq:Qqn_1}
       \left\|Q_{q_n} - \frac{1}{t_n}\int^{t_n}_0e^{is\sj_{q_n}}A_{q_n}e^{-is\sj_{q_n}}\text{d}s\right\|\leq \left(\frac{4R^2q_n^2}{t_n^2\varepsilon^2}+64R^2 e^{c_1q_n/2}\sqrt{\varepsilon}\right)^{1/2}
   \end{align}
   By Estimate II of Theorem \ref{thm:Qn-intA} applied to $Q_{q_{n+1}}$,  $t=t_n$ and $\tilde{\varepsilon}$ (value to be determined below), we have 
    \begin{align}
        \left\|Q_{q_{n+1}} - \frac{1}{t_n}\int^{t_n}_0e^{is\sj_{q_{n+1}}}A_{q_{n+1}}e^{-is\sj_{q_{n+1}}}\text{d}s\right\| \leq \left(
\frac{4R^2 q_{n+1}^2}{t_n^2\tilde\e_{}^2} + 16 R^2 q_{n+1}e^{c_1q_{n}/2}\sqrt{\tilde\e_{}}\right)^{1/2}.
    \end{align}
    Here, in order to apply Estimate II, we have to ensure that $\gamma_{n+1}\leq \tilde{\varepsilon}$. This condition will be checked below in \eqref{eq:gamma_check1}.
    We make the natural choice of 
    \begin{equation}
        t_n^2 = q_{n+1}^{}e^{-c_1{q_n}/2}\tilde\e^{-5/2}_{}
    \end{equation}
    so that
    \begin{align}\label{eq:Qqn+1_1}
        \left\|Q_{q_{n+1}} - \frac{1}{t_n}\int^{t_n}_0e^{is\sj_{q_{n+1}}}A_{q_{n+1}}e^{-is\sj_{q_{n+1}}}\text{d}s\right\| 
        & \leq \left( q_{n+1}e^{c_1q_n/2}\sqrt{\tilde\e_{}} + 16 R^2 q_{n+1}e^{c_1q_{n}/2}\sqrt{\tilde\e_{}}\right)^{1/2}\\
        & \leq 5 R \sqrt{q_{n+1}}e^{c_1q_n/4}\tilde\e^{1/4}.
    \end{align}
    We also make the choices that
    \begin{equation}
            \e =  q_n^{-2} e^{-q_n\eta_0 }, \text{ and } 
            \tilde\e_{} =q_{n+1}^{-2}e^{-q_n\eta_0 },
    \end{equation}
    with 
    \begin{align}\label{def:eta_0}
    5c_1<\eta_0<\frac{4\eta-3c_1}{5}<\eta.
    \end{align}
    Recall that $\eta>7c_1$, so the set of $\eta_0$ is non-empty.
    This choice of $\eta_0$ also ensures that 
    \begin{align}\label{eq:gamma_check1}
        \gamma_{n+1}=q_{n+1}^{-3}e^{-q_n\eta}<q_{n+1}^{-2}e^{-q_n\eta_0}=\tilde{\varepsilon}.
    \end{align}
    With these choices in place we observe $t_n^2 = q_{n+1}^{6}e^{(\frac{5}{2}\eta_0-\frac{c_1}{2})q_n}$, and hence \eqref{eq:Qqn_1} leads to:
    \begin{align}\label{eq:Qqn_2}
         &\left\|\left(Q_{q_n} - \frac{1}{t_n}\int^{t_n}_0e^{is\sj_{q_n}}A_{q_n}e^{-is\sj_{q_n}}\text{d}s\right) \phi\right\|\notag\\
         &\qquad\qquad\le \|\phi\|_{\ell^1}\left(\frac{4R^2 q_n^6}{q_{n+1}^{6}e^{(\frac{5}{2}\eta_0-\frac{c_1}{2})q_n}e^{-2\eta_0 q_n}} + 64 R^2 q_n^{-1}e^{c_1q_n/2}e^{-\eta_0q_n/2}\right)^{1/2}\notag\\
         &\qquad\qquad \leq 3R \|\phi\|_{\ell^1} e^{\frac{1}{4}(c_1-\eta_0)q_n},
    \end{align}
    and \eqref{eq:Qqn+1_1} leads to:
    \begin{align}\label{eq:Qqn+1_2}
         \left\|\left(Q_{q_{n+1}} - \frac{1}{t_n}\int^{t_n}_0e^{is\sj_{q_{n+1}}}A_{q_{n+1}}e^{-is\sj_{q_{n+1}}}\text{d}s\right) \phi\right\|\le 5R\|\phi\|_{\ell^1} e^{\frac{1}{4}(c_1-\eta_0)q_n}.
    \end{align}
     Moreover, by \cite[Theorem A.1]{Fil2017}, we have
    \begin{align}
        \frac{1}{t_n}\left\|\int^{t_n}_0e^{is\sj_{q_n}}A_{q_n}e^{-is\sj_{q_n}}-e^{is\sj_{q_{n+1}}}A_{q_{n+1}}e^{-is\sj_{q_{n+1}}}\text{d}s \right\|& \leq 2(R+1)t_n\|\sj_{n}-\sj_{n+1}\|\\
        & \leq 2(R+1)q_{n+1}^{3}e^{\frac{1}{4}(5\eta_0-c_1)q_n}\gamma_{n+1}\\
        & \leq 2(R+1)e^{\frac{1}{4}(5\eta_0 -c_1-4\eta)q_n} .
    \end{align}
    Therefore, combining the estimate above with \eqref{eq:Qqn_2} and \eqref{eq:Qqn+1_2}, we have
    \begin{align}\label{eq:Qn-Qn+1}
        \|(Q_{q_n}-Q_{q_{n+1}})\phi\|\leq (8R e^{\frac{1}{4}(c_1-\eta_0)q_n}+2(R+1)e^{\frac{1}{4}(5\eta_0-c_1-4\eta)q_n})\|\phi\|_{\ell^1}.
    \end{align}
    This is clearly summable since $\eta_0>c_1$ and $4\eta>5\eta_0-c_1$, see \eqref{def:eta_0}.  Note $Q_{q_n}$ is also uniformly bounded, $\|Q_{q_n}\|\leq 2R$, hence $Q_{q_n}$ converges strongly to a bounded self-adjoint operator $Q_{\sj}$.

    \subsection{$\bf{Q_{\sj}=\lim_{t\to \infty}t^{-1}X_{\sj}(t)}$}
    To show $\lim_{t\to\infty} t^{-1}X_{\sj_{}}(t)$ equals $Q_{\sj}$, let us fix $\phi  \in \mathfrak{D}(X)\subset \ell^1(\Z)$. Then, for a large $t>0$, let $n\in \N$ be such that $t\in [t_{n-1}, t_n]$ and 
    \[t_{n-1}=q_n^3 e^{\frac{1}{4}(5\eta_0-c_1)q_{n-1}}, \, t_n= q_{n+1}^{3}e^{\frac{1}{4}(5\eta_0-c_1)q_n}.\]
    We estimate
    \begin{align}
        \left\|\left(\frac{1}{t}X_{\sj_{}}(t) - Q_{\sj}\right)\phi \right\| & \leq \frac{1}{t}\left\|\left(X_{\sj}(t) - X_{\sj_{q_n}}(t)\right)\phi\right\| + \left\|\left(\frac{1}{t} X_{\sj_{q_n}}(t) - Q_{q_n}\right) \phi\right\| + \left\| \left(Q_{q_n} - Q_{\sj}\right)\phi\right\|. 
    \end{align}
We have shown that $\|(Q_{q_n}-Q_{\sj})\varphi\|\to 0$ as $n\to \infty$ (as $t\to \infty$). Next we estimate the other two terms separately. 
By \cite[Theorem A.1]{Fil2017}, 
\begin{align}\label{eq:XJt}
    \frac{1}{t}\left\|(X_{\sj}(t)-X_{\sj_{q_n}}(t))\phi\right\|\leq 2R t\|\sj -\sj_{q_n}\|\cdot \|\phi\|_{\ell^1}
    \leq &2R t_n q_{n+1}^{-3}e^{-\eta q_n}\|\phi\|_{\ell^1}\\
    \leq &2R e^{\frac{1}{4}(5\eta_0-c_1-4\eta)q_n}\|\phi\|_{\ell^1}\to 0,
\end{align}
since $4\eta>5\eta_0-c_1$.
Also, by Corollary \ref{cor:1}, for $\varepsilon_0>\kappa_n$, we have
    \begin{align}
        \left\|\left(\frac{1}{t}X_{\sj_{q_n}}(t)-Q_{q_n}\right)\phi\right\|& \leq \frac{1}{t_{n-1}}\|X\phi\| + \|\phi\|_{\ell^{1}}\left(\frac{4R^2 q_n^2}{t_{n-1}^2\e_0^2} + 16R^2 q_ne^{c_1q_{n-1}/2}\sqrt{\e_0}\right)^{1/2}.
    \end{align}
    Picking $\varepsilon_0=q_n^{-2}e^{-\eta_0 q_{n-1}}>q_n^{-3}e^{-\eta q_{n-1}}=\gamma_n$ (since $\eta>\eta_0$) yields
    \begin{align}
        \left\|\left(\frac{1}{t}X_{\sj_{q_n}}(t)-Q_{q_n}\right)\phi\right\|& \leq \frac{1}{t_{n-1}}\|X\phi\| +5R e^{\frac{1}{4}(c_1-\eta_0)q_{n-1}}\|\phi\|_{\ell^1}\to 0,
    \end{align}    
    since $\eta_0>c_1$. 
    Therefore $\lim_{t\in\infty} t^{-1}X_{\sj}(t)$ converges to $Q_{\sj}$ strongly.
\subsection{Nontrivial kernel}
        We are only left to show that $Q\phi\neq 0$ for every $0\neq \phi\in \mathfrak{D}(X)\cap \ell^1(\Z)$. For such a $\phi$, we normalize it so that $\|\phi\|_{\ell^2} = 1$. We have the lower bound $\|Q_{q_n }\phi\| \geq e^{-c_1q_n}$ by \cite[Equation (3.6)]{Fil2017}.
        Also for $n$ large enough, by \eqref{eq:Qn-Qn+1}, and that $\eta_0>c_1$, $4\eta>5\eta_0-c_1$,
        \begin{align}
            \sum_{k=n}^{\infty}\|(Q_{q_k}-Q_{q_{k+1}})\phi\| & \leq \sum_{k=n}^{\infty} (8R e^{\frac{1}{4}(c_1-\eta_0)q_k}+2(R+1)e^{\frac{1}{4}(5\eta_0-c_1-4\eta)q_k})\|\phi\|_{\ell^1}\\
            &\leq (8R e^{\frac{1}{4}(c_1-\eta_0)q_n}+2(R+1)e^{\frac{1}{4}(5\eta_0-c_1-4\eta)q_n})\|\phi\|_{\ell^1}.
        \end{align}
        As such,
        \begin{align}
             \|Q\phi\|
             \geq &\|Q_n\phi\|-\sum_{k=n}^{\infty}\|(Q_{q_k}-Q_{q_{k+1}})\phi\|\\
             \geq &e^{-c_1q_n}-(8R e^{\frac{1}{4}(c_1-\eta_0)q_n}+2(R+1)e^{\frac{1}{4}(5\eta_0-c_1-4\eta)q_n})\|\phi\|_{\ell^1}>0,
        \end{align}
        for $n$ large enough, since by \eqref{def:eta_0} we have $\eta_0>5c_1$ and $4\eta>5\eta_0+3c_1$. \qed

\section{A brief comparison to the quasi-periodic setting}\label{sec:QP}
The limit-periodic operator $\sj$ satisfying the condition $\sj\in EC(\eta)$ is close in spirit to, but easier than, quasi-periodic operators with Liouville frequencies, which are defined as
\begin{align}
    (H_{\omega,\theta}\phi)_n=\phi_{n-1}+v(\theta+n\omega)\phi_n+\phi_{n+1}.
\end{align}
Here we assume $v$ to be Lipschitz on $\T$ and $\omega$ is a Liouville frequency, meaning there exists a sequence of continued fraction approximants $\{p_n/q_n\}$ such that
\begin{align}
    \|q_n\omega\|_{\T}\leq e^{-\beta q_n}, 
\end{align}
for some $\beta>0$.
This means $H_{\omega,\theta}$ is exponentially close to a $q_n$-periodic operator locally, when restricted to a large box of size $e^{\beta q_n/2}$:
\begin{align}
    \left\|(H_{\omega,\theta}-H_{p_n/q_n,\theta})|_{[-e^{\beta q_n/2}, e^{\beta q_n/2}]}\right\|\lesssim e^{-\beta q_n/2}.
\end{align}
In fact this follows from that
\begin{align}
    \left\|(H_{\omega,\theta}-H_{p_n/q_n,\theta})|_{[-e^{\beta q_n/2}, e^{\beta q_n/2}]}\right\|
    \leq &\max_{|j|\leq e^{\beta q_n/2}} |v(\theta+j\omega)-v(\theta+jp_n/q_n)|\\
    \leq &\max_{|j|\leq e^{\beta q_n/2}} \mathrm{Lip}(v) \cdot \|j(\omega-p_n/q_n)\|_{\T}\\
    \leq &\max_{|j|\leq e^{\beta q_n/2}} \frac{|j|}{ q_n}\, \|q_n\omega\|_{\T}\cdot \mathrm{Lip}(v)\\
    \leq &e^{-\beta q_n/2} \mathrm{Lip}(v).
\end{align}
This is, roughly speaking, the notion of $\beta$-almost periodic in Jitomirskaya-Zhang \cite{JZ}.

In the limit-periodic setting, our condition $\sj\in EC(\eta)$ requires $\sj$ to be exponential close to a $q_n$-periodic operator globally:
\begin{align}
    \|\sj-\sj_{q_n}\|\leq q_{n+1}^3 e^{-\eta q_n}.
\end{align}

\begin{acknowledgment}
    We would like to thank Fan Yang for the helpful discussions which led to the substantial shortening the of the proof of Theorem \ref{theorem : estimate for bad sets}.
\end{acknowledgment}

\appendix

\section{Floquet Theory}
\label{appendix: Floquet Theory}

The following material is for the convenience of the reader and aimed to make the manuscript as self-contained as possible. We focus primarily on the necessary ideas and conventions of Floquet Theory which are needed to properly grasp (some of) the results in \cites{Fil2017, DaGa, DYL}, which originally motivated this paper. For a proper exposition, with many more details provided, we refer to the monograph by B. Simon \cite{Sim}.

In this manuscript our main focus is on that concerning sequences of periodic Jacobi matrices belonging to $\mathcal{J}(R)$, for $R>0$. This type of operators by definition are acting on $\ell^2(\Z)$, but can be thought as acting on any sequence, $\{\phi_n\}_{n=-\infty}^\infty$, via
\begin{equation*}
    (\sj \phi)_n = a_n\phi_{n+1}+b_n\phi_{n}+a_{n_1}\phi_{n-1}.
\end{equation*}

Suppose $\sj\in\mathcal{J}_{q_n}(R)$, that is, $\sj = \sj_{q_n}$, a $q_n$-periodic Jacobi matrix. Then it must follow that all eigenvalues of $\sj_{q_n}$ should lie on the flat torus $\T\simeq \R/2\pi\Z$. In other words, if $\phi$ is a solution for the eigenvalue problem for $\sj$, it must satisfy
\begin{equation*}
    \phi_{n+q_n} = e^{i 2\pi\theta}\phi_n, \quad \forall n\in \Z, \quad \theta \in\R.
\end{equation*}

We call such solutions \textbf{Floquet solutions}. 
\begin{definition}[Floquet Space]
    \begin{align}
        \mathcal{H}_{q_n} &\eqdef L^2\left([0,1), \C^{q_n}; \text{d}\theta\right)\\
        & =\left\{f \colon [0,1) \longrightarrow \C^{q_n } \left|\, \, \int\limits_{[0,1)}\|f(\theta)\|_{C^{q_n}}^2 \text{d}\theta < \infty\right.\right\} \\
        & = \int\limits^{\oplus}_{[0,1)} \C^{q_n}\text{d}\theta.
    \end{align}
    with an inner product given by 
    \begin{equation}
        \langle f, g\rangle_{\mathcal{H}_{q_n}} = \int_0^1\langle f(\theta), g(\theta)\rangle_{\C^{q_n}} \text{d}\theta.
    \end{equation}
    Thus, Floquet space is a Hilbert Space.
\end{definition}

\begin{definition}[Floquet transform]
    The linear operator $\mathcal{F}_{q_n} \colon \ell^{2}(\Z)\longrightarrow \mathcal{H}_{q_n}$ given by
    \begin{equation}
        \left[\mathcal{F_{q_n}\phi}\right]_{k}(\theta) = \sum_{p\in\Z}\phi_{k+pq_{n}}e^{-ip2\pi\theta}, \qquad \theta\in \T, \quad 0\leq k\leq q_n-1.
    \end{equation}
\end{definition}
\begin{remark}
    It is an easy exercise to verify that $\mathcal{F}_{q_n}$ extends unitarily from $\ell^2(\Z)$ to $\mathcal{H}_{q_n}$; where the inverse is given by
    \begin{equation}
        \mathcal{F}_{q_n}^{-1}(g_{k+pq_n}) = \int_0^1 e^{i p 2\pi\theta}g_k(\theta)\text{d}\theta, \qquad 0\leq k\leq q_n-1, \quad p\in\Z.
    \end{equation}
\end{remark}

For convenience we are going to assume that ${q_n}\geq 3$, so that for each $\theta \in \T$ we can define
\begin{equation*}
     \sj_{q_n}\left(\theta\right) \eqdef \begin{pmatrix}
        b1 & a_1 & 0 & &  \cdots & e^{-i2\pi\theta}a_{q_n}\\
        a_1 & b_2 & a_2  & & & \\
         0& a_2& b_3& \ddots & & \\
         \vdots & & \ddots & \ddots&  \ddots&\\
          & & & \ddots & b_{q_{n} -1} & a_{q_{n}-1} \\
         e^{i2\pi \theta}a_{q_n}& & & & a_{q_{n}-1} & b_{q_{n}}  \\
            \end{pmatrix}_.
\end{equation*}

Here $\sj_{q_n}(\theta)$ is a Jacobi matrix that leaves invariant the space of all \textit{Floquet solutions} of period ${q_n}$. Moreover, notice $\sj_{q_n}(\theta)$ is self-adjoint and has simple eigenvalues, that is, whenever $\theta\notin\{0, \frac{1}{2}\}$. Thus, we denote the eigenvalues of $\sj_{q_n}(\theta)$ by
\begin{equation*}
    \lambda_{ 1}(\theta) \leq \lambda_{ 2}(\theta) \leq \cdots \leq \lambda_{ {q_n}}(\theta),
\end{equation*}
accounting for multiplicities. For any operator $Q$ on $\ell^2(\Z)$ we denote its action on Floquet space by its conjugation by the Floquet transform. In our setting, with $q_n$-periodicity, this is,
\begin{equation}
\widehat Q \eqdef \mathcal{F}_{q_n}Q\mathcal{F}^{-1}_{q_n}.    
\end{equation}
\begin{proposition}\label{action on floquet space}
    Let be $\sj_{q_n}$ be a $q_n$-periodic Jacobi matrix. Then,
    \begin{equation}
        \widehat \sj_{q_n} = \int\limits_{\T}^\oplus \sj_{q_n}(\theta)\text{d}\theta,
    \end{equation}
    that is,
    \begin{equation}
        \left[\widehat \sj_{q_n }g\right](\theta) = \sj_{q_n}(\theta)g(\theta), \qquad \text{a.e. } \theta\in \T, \qquad g\in \mathcal{H}_{q_n}.
    \end{equation}
    \begin{proof}
        Let $\{e_k\}_{k=0}^{q_n-1}$ be the canonical orthonormal basis of $\C^{q_n}$ so that $\{e^{-2\pi ip\theta}e_k \colon 0\leq k\leq q_{n}-1, \, p\in\Z\}$ is an orthonormal basis of $\mathcal{H}_{q_n}$. Thus, verifying on any basis element,
        \begin{align}
            \widehat \sj e^{-2\pi ip\theta}e_{k} & = \mathcal{F}_{q_n}\sj \mathcal{F}^{-1}_{q_n}e^{-2\pi ip\theta} e_{k}\\
            & = \mathcal{F}_{q_n}\sj e_{k}\\
            & =\mathcal{F}_{q_n}(a_{k-1}e_{k-1}+ b_ke_k+a_{k}e_{k+1})\\
            & = a_{k-1}e^{-2\pi ip\theta}e_{k-1} + b_ke^{-2\pi ip\theta}e_{k} +a_ke^{-2\pi ip\theta}e_{k+1}\\
            & = e^{-2\pi ip\theta}(a_{k-1}e_{k-1}+ b_ke_k+a_{k}e_{k+1}) \\
            & = \sj_{q_n}(\theta) \left[e^{-2\pi ip\theta}e_k\right],
        \end{align}
        where the cases $k=0, q_n-1$ follow by taking an appropriate re-labeling. 
    \end{proof}
\end{proposition}

\section{The position operator and quantum dynamics}\label{appendix: B}

Let $X$ be the position operator as in \eqref{position operator}, and $\sj_{q_n}\in \mathcal{J}_{q_n}(R)$ be a $q_n$-periodic Jacobi matrix. We define the \textit{momentum operator} induced by $\sj_{q_n}$ by,
\begin{equation}
    A_{q_n} \eqdef i[\sj_{q_n}, X] = i\sj_{q_n} X - i X\sj_{q_n}.
\end{equation}
Applying \ref{action on floquet space} to $A_{q_n}$ we see that its action over Floquet space depends on the angle $\theta\in \T$. Further, we see that if $\sj_{q_n}(\theta)$ has $q_n$-periodic boundary conditions, we must have that $A_{q_n}(\theta)$ should be of the form
\begin{equation}\label{eq:Aqn_theta}
    A_{q_n}(\theta) = \begin{pmatrix}
        0 & ia_1 & 0 & &  \cdots & -ie^{-2\pi i\theta}a_{q_n}\\
        -ia_1 & 0 & ia_2  & & & \\
         0& -ia_2& 0& \ddots & & \\
         \vdots & & \ddots & \ddots&  \ddots&\\
          & & & \ddots & 0 & ia_{q_{n}-1} \\
         ie^{2\pi i \theta}a_{q_n}& & & & -ia_{q_{n}-1} & 0  \\
            \end{pmatrix}_.
\end{equation}

Clearly, $A_{q_n}(\theta)$ is self-adjoint. Thus, similarly to the previous discussion about $q_n$-periodic Jacobi matrices we observe
\begin{align}
    \mathcal{F}_{q_n}A_{q_n}\mathcal{F}^{-1}_{q_n}e^{-2 \pi i p\theta}e_{k} & = \mathcal{F}_{q_n}\left(i\sj_{q_n} X - i X\sj_{q_n}\right)\mathcal{F}^{-1}_{q_n}e^{-2 \pi i p\theta}e_{k}\\
    & = i\mathcal{F}_{q_n}\left(\sj_{q_n} X \right)\mathcal{F}^{-1}_{q_n}e^{-2 \pi i p\theta}e_{k} - i \mathcal{F}_{q_n}\left( X\sj_{q_n}\right)\mathcal{F}^{-1}_{q_n}e^{-2 \pi i p\theta}e_{k}\\
    & = i\left( ka_{k-1}e^{-2\pi ip\theta}e_{k-1} + kb_ke^{-2\pi ip\theta}e_{k} +ka_ke^{-2\pi ip\theta}  e_{k+1}\right)  \\
    & - i\left( (k-1)a_{k-1}e^{-2\pi ip\theta}e_{k-1} + kb_ke^{-2\pi ip\theta}e_{k} +(k+1)a_ke^{-2\pi ip\theta}e_{k+1}\right)\\
    & =e^{-2\pi ip\theta}\left(i a_{k-1}e_{k-1} - ia_ke_{k+1}\right)\\
    & = A_{q_n}(\theta)e^{-2\pi ip\theta}e_{k}
\end{align}
In light of the identification given by Proposition \ref{action on floquet space}, we can regard the action of a given Jacobi matrix, $\sj_{q_n}$, as matrix multiplication, by $\sj_{q_n}(\theta)$, on each fiber over the flat-torus. Thus, we can study the dynamical properties of a Jacobi matrix by instead focusing on the properties of a continuous family of period-preserving finite matrices $\{\sj_{q_n}(\theta)\}_{\theta\in\T}$. 

Similar as Fillman \cite{Fil2017}--- and following closely Damanik-Yessen's argument --- our end goal is to be able to relate the Heisenberg evolution associated to the Jacobi matrix $\sj$, to some operator $Q$, which depends on the angle $\theta$. To do so, by the previous discussion, we look at the measurable sections. These are, maps of the form
\begin{equation}
    e_k \colon \T \longrightarrow \bigsqcup_{\theta \in \T} \C^{q_n}_\theta
\end{equation}
which, morally speaking, assigns to (almost) every $\theta\in \T$ an orthonormal basis $\{e_k(\theta) \,\colon\, 0\leq k \leq q_n-1\}$ of $\C^{q_n}_\theta $. We can think of $\sj_{q_n}$ as acting as matrix multiplication at each fixed fiber $\theta\in \T$. Thus, the eigenvalue problem at each fiber is given by 
\begin{equation}
    \sj(\theta)_{q_n}(\theta)e_k(\theta) = \lambda_{q_n,k}(\theta)e_k(\theta), \quad\quad \forall 0\leq k \leq q_n-1, \quad \theta\in\T. 
\end{equation}

Recall that $\lambda_{q_n,k}(\theta)$ is analytic in $\T\setminus\{0,\frac{1}{2}\}$ (see, \ref{theorem: properties of the eigenvalues}.) Let
\begin{equation}
    P_k(\theta)\colon \int\limits_\T^{\,\,\oplus }C^{q_n}_\theta \text{d}\theta \longrightarrow \ker\left(\sj_{q_n}(\theta) -\lambda_{q_n,k}(\theta)\mathbb{I}\right),
\end{equation}
where $\mathbb{I}$ is the identity matrix on each $\C^{q_n}_{\theta}$.

By the proof of Theorem 1.6 of Damanik-Lukiç-Yessen in \cite{DYL}, the $Q$ operator in  can by identified using the direct integral formalism, that is,
\begin{equation}
    \widehat{Q}_{q_n} = \int\limits_{\T}^{\,\,\oplus} Q_{q_n}(\theta) \text{d}\theta.
\end{equation}

Here, 
\begin{equation}
        Q_{q_n}(\theta ) = q_n \sum_{k=0}^{q_n-1}\dot{\lambda}_{q_n,k}(\theta)P_{k}(\theta), \qquad \theta\in \T\setminus\left\{0,\frac{1}{2}\right\};
    \end{equation}
with the observation,
\begin{equation}
    \sigma(\sj_{q_n}) = \bigcup_{\theta\in\T}\sigma(\sj_{q_n}(\theta)) = \bigcup_{\theta\in\T}\bigcup_{k=0}^{q_n-1}\{\la_{q_n,k}(\theta)\}.
\end{equation}


\begin{thebibliography}{99}

    \bibitem{chau} Chulaevsky V.: Almost Periodic Operators and Related Nonlinear Integrable Systems. Manchester University Press, Manchester (1989)

    \bibitem{DaFi}Damanik, D. and Fillman, J.: One-Dimensional Ergodic Schrödinger Operators: II. Specific Classes (Vol. 249). American Mathematical Society (2025).
    
    \bibitem{DaGa} Damanik D., Gan Z.: Spectral properties of limit-periodic Schrödinger operators. Commun. Pure Appl. Anal. (3) 10, 859–871 (2011)
    
    \bibitem{DYL} Damanik D., Lukic M., Yessen W.: Quantum dynamics of periodic and limit-periodic Jacobi and block Jacobi matrices with applications to some quantum many body problems. Commun. Math. Phys. 337, 1535–1561 (2015)

    \bibitem{Fil2017} Fillman, J. Ballistic Transport for Limit-Periodic Jacobi Matrices with Applications to Quantum Many-Body Problems. Commun. Math. Phys. 350, 1275–1297 (2017)

    \bibitem{HPS}Han, R., Powell, L., and Shin, H. Graph indexes for lattices with torodial boundary conditions. Preprint.

    \bibitem{JZ}Jitomirskaya, S. and Zhang, S., 2022. Quantitative continuity of singular continuous spectral measures and arithmetic criteria for quasiperiodic Schrödinger operators. Journal of the European Mathematical Society (EMS Publishing), 24(5).

    \bibitem{Ka}Kachkovskiy, I., 2016. On transport properties of isotropic quasiperiodic XY spin chains. Communications in Mathematical Physics, 345(2), pp.659-673.
    
    \bibitem{Kat} Kato T.: Perturbation Theory for Linear Operators, Springer-Verlag, Berlin, Heidelberg, New York (1966).

    \bibitem{Last} Last, Y. On the measure of gaps and spectra for discrete 1D Schrödinger operators. Commun.Math. Phys. 149, 347–360 (1992).
    \bibitem{La2}Last, Y., 1996. Quantum dynamics and decompositions of singular continuous spectra. Journal of Functional Analysis, 142(2), pp.406-445.

    \bibitem{R&S} Reed M., Simon B.: Methods of modern mathematical physics. I. Functional Analysis. Academic Press [Harcourt Brace Jovanovich Publishers], New York (1980)
    
    \bibitem{Sim} Simon, B.: Szegő’s theorem and its descendants. Spectral theory for $L^2$ perturbations of orthogonal polynomials. M. B. Porter Lectures. Princeton University Press, Princeton (2011)

    \bibitem{Tesc} Teschl, G.: Jacobi operators and completely integrable nonlinear lattices. In: Mathematical Surveys and Monographs, vol. 72. American Mathematical Society, Providence (2000)

    \bibitem{Tesc2} Teschl, G.: Mathematical methods in quantum mechanics. American Mathematical Society,  Providence, (2014). 

    \bibitem{Young}Young, G., 2023. Ballistic transport for limit-periodic Schrödinger operators in one dimension. Journal of Spectral Theory, 13(2), pp.451-489.

    \bibitem{ZZ}Zhang, Z. and Zhao, Z., 2015. Ballistic Transport and Absolute Continuity of One-Frequency Schr\"{o} dinger Operators. arXiv preprint arXiv:1512.02195.

    \bibitem{Zh}Zhao, Z., 2016. Ballistic motion in one-dimensional quasi-periodic discrete Schrödinger equation. Communications in Mathematical Physics, 347(2), pp.511-549.
    
\end{thebibliography}
\end{document}